\documentclass[conference,letterpaper,11pt]{IEEEtran}

\usepackage{todonotes}
\usepackage[utf8]{inputenc} 
\usepackage[T1]{fontenc}
\usepackage{url}
\usepackage{ifthen}
\usepackage{cite}
\usepackage[cmex10]{amsmath} 

\usepackage{amsfonts,amssymb,amsmath,amsthm,amstext}

\newcounter{note}[section]

\newcommand{\venkat}[1]{\refstepcounter{note}$\ll${\bf Venkat~\thenote:}
  {\sf #1}$\gg$\marginpar{\tiny\bf VG~\thenote}}

\newtheorem{theorem}{Theorem}[section]

\newtheorem{conjecture}[theorem]{Conjecture}

\newtheorem{remark}[theorem]{Remark}

\newtheorem{defn}[theorem]{Definition}

\newcommand{\F}{\mathbb{F}}

\newcommand{\eat}[1]{}

\newcommand{\tr}{\mathrm{tr}}
\newcommand{\calC}{\mathcal{C}}
\newcommand{\calR}{\mathcal{R}}
\newcommand{\calP}{\mathcal{P}}
\newcommand{\dual}{\calC^{\bot}}

\begin{document}

\title{Near-optimal Repair of Reed-Solomon Codes with Low Sub-packetization}

\author{%
  \IEEEauthorblockN{Venkatesan Guruswami}
  \IEEEauthorblockA{
		Computer Science Department \\
		Carnegie Mellon University \\
                Email: venkatg@cs.cmu.edu}
\and
  \IEEEauthorblockN{Haotian Jiang}
  \IEEEauthorblockA{Paul G. Allen School of Computer Science \& Engineering \\
  University of Washington \\
  Email: jhtdavid@uw.edu}
  }

\maketitle

\begin{abstract}
  Minimum storage regenerating (MSR) codes are MDS codes which allow for recovery of any single erased symbol with optimal repair bandwidth, based on the smallest possible fraction of the contents downloaded from each of the other symbols. Recently, certain Reed-Solomon codes were constructed which are MSR. However, the sub-packetization of these codes is exponentially large, growing like $n^{\Omega(n)}$ in the constant-rate regime. In this work, we study the relaxed notion of $\epsilon$-MSR codes, which incur a factor of $(1+\epsilon)$ higher than the optimal repair bandwidth, in the context of Reed-Solomon codes. We give constructions of constant-rate $\epsilon$-MSR Reed-Solomon codes with polynomial sub-packetization of $n^{O(1/\epsilon)}$ and thereby giving an explicit tradeoff between the repair bandwidth and sub-packetization.
  
\end{abstract}



\section{Introduction}
In practical distributed storage systems, data is stored in encoded form on a large number of individual storage nodes to protect it from node failure.
The original file is encoded using an $[n,k]$ code and is distributed among $n$ different storage units. Each one of the storage units is called a \emph{node} and contains one symbol of the codeword.
Maximum Distance Separable (MDS) codes are widely used in this case because they achieve the highest possible error-correction capability for a given redundancy level --- the codeword can be 
fully recovered by accessing any $k$ of its $n$ entries.
Reed-Solomon (RS) codes are among the most commonly used MDS codes, including in storage applications.


In practice, however, the most common scenario is the failure of a single node~\cite{rashmi2015hitchhiker}.
To repair a single node failure, the system needs to download information from some other \emph{helper} nodes of the codeword, and the total amount of downloaded data is called the \emph{repair bandwidth}.
By the MDS property, it is possible to recover the content of the entire file, and therefore that of the failed node, by downloading the entire content stored on any $k$ nodes. (On the other hand, the same MDS property means that one cannot recover any information about the failed node if one contacts fewer than $k$ nodes.)
However, this naive approach can be far from optimal and it is shown in~\cite{dimakis2010network} that one can save repair bandwidth by contacting $d > k$ helper nodes and downloading only \emph{part} of the information stored on each of those helper nodes. Codes with this property are called regenerating codes.
In fact, one can attain the maximum savings by contacting all other $n-1$ nodes (i.e. $d=n-1$).
We shall therefore focus on the case when $d = n-1$ throughout this paper.

Initial constructions of regenerating codes were \emph{vector} MDS codes, whose symbols are vectors in $\F^l$ for some field $\F$, with the codes being $\F$-linear. The quantity $l$ is called the \emph{sub-packetization} of the code.  We focus on \emph{linear repair schemes} where the helper nodes then return (much fewer than $l$) $\F$-linear combinations of the vector stored at them.


Shanmugam et al.~\cite{shanmugam2014repair} first propose a framework for studying the repair bandwidth of \emph{scalar} MDS codes. In this framework, each node contains a symbol of some finite symbol field $E$, which is a degree $l$ extension of some base field $\F_q$ (i.e. $[E:\F_q] = l$). The code itself is a linear MDS code over the bigger field $E$.
The symbol field $E$ can thus be viewed as a $l$-dimensional vector space on $\F_q$ and the code can also be viewed as a vector code over $\F_q$ with sub-packetization $l$. 
When a helper node is contacted, instead of returning a symbol of the symbol field $E$, it returns sub-symbols of the base filed $\F_q$.
The repair bandwidth is formally defined as the total amount of sub-symbols of $\F_q$ downloaded from all the helper nodes.

\begin{defn}[Repair bandwidth]
Let $\calC$ be an $[n,k]$ MDS code with sub-packetization $l$ over a finite subfield $\F_q$. For $i \in \{ 1,2 \cdots, n\}$ and $\calR \subseteq [n] \backslash \{i\}$, define $N(\calC,i,\calR)$ as the smallest number of sub-symbols of $\F_q$ (which can be linear combinations of the sub-symbols stored in the helper nodes) one needs to download from the helper nodes $\{c_j: j \in \calR \}$ in order to repair the failed node $c_i$. The repair bandwidth of the code $\calC$ is defined as
\[
\max_{ i \in [n]} N(\calC, i,[n] \backslash \{i\})
\]
\end{defn}
For any MDS code, Dimakis et al.~\cite{dimakis2010network} provides a benchmark for the repair bandwidth by giving an achievable lower bound, known as the \emph{cut-set bound}. 
Codes that achieve this lower bound are known as Minimum Storage Regenerating (MSR) codes.  

\begin{defn}[Cut-set bound and MSR codes]
\label{defn:cut-setbound}
Let $\calC$ be an $[n,k]$ (scalar or vector) MDS code with sub-packetization $l$ over some base field $\F_q$. For any $i \in [n]$ and any subset $\calR \subseteq [n] \backslash \{i\}$, we have the following inequality:
\[
N(\calC,i,\calR) \geq \frac{(n-1)l}{n-k} \ ,
\]
with equality attained if and only if each node in $[n] \backslash \{i\}$ returns $\frac{l}{n-k}$ symbols of $\F_q$. A code achieving this lower bound is called a minimum storage regenerating (MSR) code.
\end{defn}
By now several constructions of MSR codes are known. However, the sub-packetization of these constructions is large. In particular, the constructions in \cite{YeB17b,SVK16} are explicit with small field size, and achieve sub-packetization $l \approx r^{n/r}$. Such a large sub-packetization has been shown to be inherent to MSR codes, with a lower bound of $\exp(\Omega(\sqrt{n/r}))$~\cite{goparaju2014improved} which was recently improved to a near-optimal  $\exp(\Omega(n/r))$ lower bound~\cite{AlrabiahGuruswami18}.

While essentially optimal MSR codes have now been constructed, these are tailormade constructions of \emph{vector} MDS codes. It is of significant theoretical and practical interest to study whether (scalar) MDS codes like Reed-Solomon codes, which are already widely used in practice, can allow for efficient regeneration of a failed node. There have been several recent works in this direction. Guruswami and Wootters~\cite{guruswami2017repairing} give an exact characterization for linear repair schemes of scalar MDS codes using dual codes. They show that to obtain a low repair bandwidth, it suffices to find a set of dual codewords that span the entire field $E$ on the failed node, but have low dimension on other nodes.

\begin{defn}[Dual Code]
  The dual code of a linear code $\calC \subseteq E^n$ is the linear subspace of $E^n$ defined by $\dual = \{ x = (x_1,\cdots,x_n) \in E^n  \mid \sum_{i=1}^n x_i c_i = 0 \ \forall c=(c_1,\cdots,c_n) \in \calC \}$.
  \end{defn}

\begin{theorem}[~\cite{guruswami2017repairing}]
Let $C \subseteq E^n$ be a scalar linear MDS code of length $n$. Let $\F_q$ be a subfield of $E$ such that $[E:\F_q] = l$. For a given $i \in \{1,\cdots,n\}$, the following statements are equivalent.\\
(1) There is a linear repair scheme of node $c_i$ over $\F_q$ such that the repair bandwidth $N(\calC,i,[n]\backslash \{i\}) \leq b$.\\
(2) There is a subset of codewords $\calP_i \subseteq \dual$ with size $|\calP_i| = l$ such that
\[
\dim_{\F_q}(\{x_i : x\in \calP_i\}) = l,
\]
and
\[
b \geq \sum_{j \in [n] \backslash \{i\}} \dim_{\F_q}(\{x_j: x \in \calP_i\})
\]
\end{theorem}

Using their characterization, the authors of~\cite{guruswami2017repairing} constructs a family of RS codes with low sub-packetization $l=\log_{n/r} n$ and a repair scheme with optimal repair bandwidth in this regime. However, the repair bandwidth is much higher than the cut-set bound, which can only be achieved for large sub-packetization.
Subsequent work~\cite{dau2017optimal,dau2018repairing} generalize the results in~\cite{guruswami2017repairing} but none of their results approaches the cut-set bound.

In a beautiful work, Tamo, Ye, and Barg~\cite{tamo2017optimal} constructed Reed-Solomon codes that are MSR, i.e., admit repair schemes with repair bandwidth meeting the cut-set bound. The sub-packetization is huge, $l = n^{O(n)}$, but they also prove a lower bound of $k^{\Omega(k)}$ for scalar MDS codes, which is even higher than the exponential lower bound for general MSR (vector MDS) codes.

Given the large sub-packetization of MSR codes which is not suitable for practical applications, Guruswami and Rawat~\cite{GR-soda17} proposed and studied codes that trade-off repair bandwidth with sub-packetization, They constructed codes with sub-packetization $l$ as small as $r=n-k$ with repair bandwdith at most twice the cut-set bound, and $l\approx r^{1/\epsilon}$ with repair bandwidth at most $(1+\epsilon)$ times the cut-set bound, i.e., bounded by $(1+\epsilon)(n-1)l/r$. In a later work, Rawat {\it et al}~\cite{RTGE} propose $\epsilon$-MSR codes where the download from \emph{each helper node} is at most $(1+\epsilon)l/r$ (so there is also load balancing across nodes). They also construct $\epsilon$-MSR codes with sub-packetization of $r^{O(r/\epsilon)} \log n$ by combining short MSR codes with long codes of large relative minimum distance. 
\begin{defn}[$\epsilon$-MSR code]
Let $\calC$ be an $[n,k]$ (scalar or vector) MDS code with sub-packetization $l$ over some base field $\F_q$. It is said to be $\epsilon$-MSR if for every  $i \in [n]$, we have $N(\calC,i,[n]\setminus \{i\}) \leq (1+\epsilon) \cdot \frac{(n-1)l}{n-k}$, with each node returning at most $(1+\epsilon) \cdot \frac{l}{n-k}$ symbols of $\F_q$ during the repair process.
\end{defn}

Given the recent developments on MSR Reed-Solomon codes (with large sub-packetization) and $\epsilon$-MSR codes (with low sub-packetization), a natural question that arises is whether we can 
combine the benefits of both these lines of work, and obtain $\epsilon$-MSR Reed-Solomon codes with low sub-packetization. This is precisely the question addressed in this work. In this paper, we provide a partial answer to this question by constructing two families of RS codes that achieve small repair bandwidth using polynomial sub-packetization in the constant rate regime of $k = \Theta(n)$. (Our constructions also work beyond this regime, but we will be focusing on the tradeoff in the constant rate regime for simplicity.)

The constructions in this paper rely on the technique of picking multiple prime numbers introduced in~\cite{tamo2017optimal}.
\begin{itemize}
\item Our first construction (Section~\ref{subsec:O(1)-MSR}) gives a family of $[n,k]$ $O(1)$-MSR RS codes with sub-packetization $O(n-k)^{O(1)}$.
\item A more careful choice of parameters leads to our second construction (\ref{subsec:O(eps)-MSR}) of a family of $[n,k]$ $\epsilon$-MSR RS codes with sub-packetization $O(n-k)^{O(1/\epsilon)}$.
\end{itemize}

Moreover, we conjecture that this tradeoff is essentially tight.
\begin{conjecture}[Tradeoff between repair bandwidth and sub-packetization]
Any $[n,k]$ $\epsilon$-MSR RS code has sub-packetization $(n-k)^{\Omega(1/\epsilon)}$ and this is tight up to a constant factor in the exponent.
\end{conjecture}

\begin{remark}
In our constructions, the number of helper nodes from which one needs to download information to repair the failed node might be smaller than $n-1$. Nevertheless, we are comparing our repair bandwidth with the cut-set bound in Definition~\ref{defn:cut-setbound} where the number of helper nodes is $n-1$, which is the smallest possible.
\end{remark}

\noindent \textbf{Related independent work on near-MSR RS codes.}
We notice that the question of understanding the tradeoff between repair bandwidth and sub-packetization for RS codes is also studied in a recent independent work~\cite{li2017tradeoff}. 
Li et al.~\cite{li2017tradeoff} give four different constructions of RS codes using three different schemes.
Their constructions, however, work in very different regimes and do not admit as clean and explicit tradeoff as our results.
To get an idea of how their results compare with ours, we assume $k = r = n/2$ which is a representative case of the constant rate regime of $k = \Theta(n)$.
Their first construction achieves repair bandwidth of $\frac{l}{a}(n-1)(a-s)$ when $n < q^a, r = n-k > q^s$ and $a|l$. 
This construction saves a factor of about $\frac{a-s}{2a} \geq \Omega\left(\frac{1}{\log_q n}\right)$ from the naive repair scheme but is still far away from the cut-set bound, which is only $O(1/n)$ times the bandwidth of the naive repair scheme.
This construction, however, works for logarithmically small sub-packetization $l = O(\log_q n)$ while our construction is for polynomially large $l$.
Their second construction saves almost half of the repair bandwidth compared to the naive scheme when $l \approx r\log_q n$.
This is again for small sub-packetization.
Their third construction attains the repair bandwidth of $\frac{l}{r}(n+1+r(q^a-2))$ when $n \leq (q^a -1) \log_r\frac{l}{a}$ and $a|l$.
In order for this to be within a constant factor from the cut-set bound, one needs to have $q^a - 2 = O(1)$ which implies an exponentially large sub-packetization $ \frac{l}{a} \geq r^{\Omega(n)} = \Omega(n)^{\Omega(n)}$. 
This is sub-optimal compared to our construction which only needs polynomially large sub-packetization.
The last construction in~\cite{li2017tradeoff} attains the repair bandwidth of $\frac{l}{r}(n-1 + (r-1)(q^a -2))$ when $n \leq (q^a - 1)m$ and $l/a \approx m^m$. 
Again if we are targeting at $O(1)$-MSR RS codes, we need to take $q^a - 2 = O(1)$ and this gives $m = \Omega(n)$ and an exponentially large sub-packetization $\frac{l}{a} = \Omega(n)^{\Omega(n)}$, which is far from optimal compared to our results.

\section{Constructions of Near-MSR Reed-Solomon Codes}
\subsection{Notation}
Throughout this paper, for positive integer $i$, we use $[i]$ to denote the set $\{1,\cdots,i\}$. 
$[n,k]$ is used to denote the length and the dimension of a code.
$r := n-k$ is the number of parity symbols of a code.
We use $E$ to denote the finite symbol field of the code, i.e. each coordinate of a codeword is a symbol in $E$.
We use $\F_q$ to denote the base field and $l$ the sub-packetization of $E$ over $\F_q$, i.e. $[E:\F_q] = l$.
When $F$ is a degree $t$ extension of a base field $\F_q$, the trace function $\tr_{F/\F_q}:F \rightarrow \F_q$ is defined as 
\[
\tr_{F/\F_q}(x) := x + x^{q} + x^{q^2} + \cdots + x^{q^{t-1}}. 
\]

\subsection{A family of $O(1)$-MSR Reed-Solomon codes with $O(r)^{O(1)}$ Sub-packetization}
\label{subsec:O(1)-MSR}
In this section, we give the construction of a family of Reed-Solomon codes with polynomial sub-packetization whose repair bandwidth is within a small constant times the cut set bound.
The construction given below achieves the smallest sub-packetization that can be achieved using our construction.
\begin{theorem}
\label{thm:O(1)-MSR}
Suppose $r :=  n-k =  cn$, where $c \in (0,1)$ is an arbitrary constant. There exists a family of $(n,k)$ Reed-Solomon codes for sufficiently large $n$ with sub-packetization at most $(\frac{r}{2})^{\lceil \frac{2+\delta}{c} \rceil}$, for some small constant $\delta > 0$, such that any code in this family is $((2+\delta)(1-c/2)-1)$-MSR.
\end{theorem}
\begin{proof}
We start by picking $m =\lceil \frac{2+\delta}{c} \rceil$ primes
$p_1,\cdots,p_m$ in the range $[r/(2+\delta), r/2]$. This can be done
for sufficiently large $n$ (and therefore sufficiently large $r$) since the number of primes in $\{1,2,\cdots,N\}$ is roughly $\frac{N}{\log N}$ according to the prime number theorem and $m$ is only a constant. For any $i \in [m]$, let $\alpha_i$ be an
element with degree $p_i$ over the base field $\F_q$.  Denote $E =
\F_q(\alpha_1,\cdots,\alpha_m)$ and $F_i = \F_q(\{\alpha_j, j\neq
i\})$.  It follows that $E=F_i(\alpha_i)$ and $[E:F_i] = p_i$.  Let
$S_i$ be the set of conjugates of $\alpha_i$, i.e. $S_i := \{\alpha_i,
\alpha_i^q, \cdots, \alpha_i^{q^{p_i-1}}\}$.
%
Since $|\cup_{i \in [m]} S_i| \geq n$, we can pick $n$ elements from $\cup_{i \in [m]}S_i$ as
the set of evaluation points.  Denote the evaluation set by
$S=\{\beta_1,\cdots,\beta_n\}$ and the corresponding $(n,k)$
Reed-Solomon code as $\calC$.

Now we show how to repair a node corresponding to evaluation point
$\beta_{i^*} \in S_i$ for a codeword $c =(c_1,\cdots,c_n) \in \calC$.
We start by picking $p_i + k-1$ evaluation points from $S\backslash
S_i$ as our repair set $R_{i^*}$.  This can be done since
$|S\backslash S_i| \geq n - r/2 = r/2 + k \geq p_i + k-1$.  Now
consider the polynomial
$$
h(x) = \prod_{j \in [n] \backslash (R_{i^*} \cup \{i^*\})} (x-\beta_j)
$$
which vanishes at each node other than those corresponding to evaluation points in $R_{i^*} \cup \{i^*\}$.
For any $s \leq p_i-1$, the polynomial $x^s h(x)$ has degree at most $n - k -1$.
Notice that the dual code $\dual$ is a $(n,n-k)$ generalized Reed-Solomon code with some coefficients $(v_1,\cdots,v_n)$ and evaluation points $S$. 
Thus we have $(v_1\beta_1^s h(\beta_1), \cdots, v_n \beta_n^s h(\beta_n)) \in \dual$, i.e.
$$
v_{i^*} \beta_{i^*}^s h(\beta_{i^*}) c_{i^*} =  - \sum_{j \neq {i^*}} v_j h(\beta_j) \beta_j^s c_j
$$
Denote the trace $\tr_{E/F_i}$ as $\tr_i$. 
We take $\tr_i$ on both sides in the equation above
\begin{align*}
\tr_{i^*}(v_{i^*} \beta_{i^*}^s h(\beta_{i^*}) c_{i^*}) &= -\sum_{j \neq {i^*}} \tr_i(v_j \beta_j^s h(\beta_j) c_j) \\
&= - \sum_{j \neq {i^*}} \beta_j^s \tr_i(v_j h(\beta_j) c_j)
\end{align*}
where in the last equality, we use the fact that $\beta_j \notin S_i$ and therefore $\beta_j \in F_i$.
Since $v_{i^*} h(\beta_{i^*}) \neq 0$ and $\beta_{i^*} \in S_i$, $\{v_{i^*} \beta_{i^*}^s h(\beta_{i^*}) \}_{s=0,\cdots, p_i-1}$ forms a basis of $E$ over $F_i$.
Thus $c_{i^*}$ can be reconstructed from $\{ \tr_i(v_{i^*} \beta_{i^*}^s h(\beta_{i^*}) c_{i^*})\}_{s=0,\cdots, p_i-1}$.

For each node $j \neq i^*$, we only need to download one symbol $\tr_i(v_j h(\beta_j) c_j)$ from the field $F_i$, which is $1/p_i$ the number of symbols stored in node $j$.
Thus the number of symbols downloaded from each node is at most 
$$
\frac{p_i + k-1}{p_i} \cdot l \leq (2+\delta)(1-c/2) \cdot \frac{n-1}{r} \cdot l,
$$
implying that the code is $((2+\delta)(1-c/2)-1)$-MSR.
We finish the proof of the theorem by noticing that the sub-packetization is exactly 
\[
\prod_{i \in [m]} p_i \leq (r/2)^{\lceil \frac{2+\delta}{c} \rceil}. \qedhere
\]
\end{proof}

\subsection{A family of $\epsilon$-MSR Reed-Solomon codes with $O(r)^{O(1/\epsilon)}$ Sub-packetization}
\label{subsec:O(eps)-MSR}
In this section, we give a tradeoff between sub-packetization and bandwidth when the bandwidth is close to the cut-set bound. 
Specifically, when the repair bandwidth is within $(1+\epsilon)$ times the cut-set bound, our construction has sub-packetization at most $O(r)^{O(1/\epsilon)}$.  
More specifically, we have the following theorem.
\begin{theorem}
\label{thm:O(eps)-MSR}
Suppose $r := n-k =  cn$, where $c \in (0,1)$ is some arbitrary constant. For any constant $\epsilon>0$, there exists a family of $[n,k]$ Reed-Solomon codes for sufficiently large $n$ with sub-packetization at most 
\[
\left( r\cdot (1-c_1\epsilon) \right)^{\lceil \frac{1}{c c_1 \epsilon} \rceil}
\]
where 
\[
c_1 := \frac{1}{\delta+ \max\{\epsilon+1-c, 2\epsilon\}}
\]
for some small positive constant $\delta$, such that any code in this family is $\epsilon$-MSR.
\end{theorem}
\begin{proof}
Define $c_2 := \frac{1}{\max\{\epsilon+1-c,2\epsilon\}} > c_1$ to be some slightly larger constant and $m := \lceil \frac{1}{c c_1 \epsilon} \rceil$.
We start by picking $m$ prime numbers $p_1, \cdots, p_m$ in the range $[r (1 -  c_2\epsilon), r (1 - c_1 \epsilon)]$, which can always be done when $n$ is large enough according to the prime number theorem.
Compared with the construction in Theorem~\ref{thm:O(1)-MSR}, we are picking a factor of $O(1/\epsilon)$ more prime numbers in a narrower range when $\epsilon$ is small.
It increases the sub-packetization by a factor of $O(1/\epsilon)$ on the exponent but allows us to pick only $O(\epsilon n)$ evaluation points corresponding to each prime number.
In this way, when a node that corresponds to some prime number $p_i$ crashes, we can pick more helper nodes that correspond to different primes and thus decrease the repair bandwidth.

For each $i \in [m]$, denote $\alpha_i$ an element with degree $p_i$ over the base field $\F_q$ and let $S_i := \{\alpha_i, \alpha_i^q, \cdots, \alpha_i^{q^{p_i-1}}\}$.
%
Again denote $E = \F_q(\alpha_1,\cdots,\alpha_m)$ and $F_i = \F_q(\{\alpha_j, j\neq i\})$. 
For each $i$, we pick a set of $S'_i \subset S_i$ with $\lceil c_1 \epsilon r \rceil$ elements.
This can be done because $|S'_i| = \lceil c_1 \epsilon r \rceil \leq r(1 - c_2 \epsilon) \leq p_i$.
Since $|\cup_{i \in [m]} S'_i| \geq m \cdot \lceil c c_1 \epsilon n \rceil \geq n$, we can pick a set of $n$ evaluation points in $\cup_{i \in [m]} S'_i$.
Denote the evaluation set by $S=\{\beta_1,\cdots,\beta_n\}$ and the corresponding $(n,k)$ Reed-Solomon code as $\calC$.

The way to repair a single node failure is similar to the one used in the proof of Theorem~\ref{thm:O(1)-MSR}. 
When a node corresponding to $\beta_{i^*} \in S_i$ fails, we will pick a set of $(p_i + k -1)$ nodes from $S \backslash S_i$ to transmit information to the failed node and apply the same repair scheme as in the proof of theorem~\ref{thm:O(1)-MSR}.
This set of back-up nodes can be chosen since 
\[
p_i + k-1 \leq n - \lceil c_1 \epsilon r \rceil
\]
where the RHS is exactly the number of nodes not in $S_i$.
The number of symbols downloaded from each helper node towards repair of a node in $S_i$ is thus 
\begin{align*}
\frac{p_i + k -1 }{p_i}\cdot l &\leq \frac{n -1 - c c_2 \epsilon n}{r - c c_2 \epsilon n} \cdot l\\
&\leq (1+\epsilon) \cdot \frac{n-1}{r} \cdot l \ ,
\end{align*}
implying that the code is $\epsilon$-MSR. 
The sub-packetization equals
\[
\prod_{i \in [m]} p_i \leq
\left( r\cdot (1-c_1\epsilon) \right)^{\lceil \frac{1}{c c_1 \epsilon} \rceil}\ . \qedhere
\]
\end{proof}

\bibliographystyle{plain}
\bibliography{bib}

\end{document}